\documentclass[aps,a4paper,onecolumn,11pt,unpublished]{quantumarticle}
\pdfoutput=1

\usepackage[utf8]{inputenc}
\usepackage[UKenglish]{babel}
\usepackage[T1]{fontenc}

\usepackage[usenames,dvipsnames]{xcolor}
\usepackage{graphicx}
\usepackage[ colorlinks = true, 
             linkcolor = MidnightBlue,
             urlcolor  = MidnightBlue,
             citecolor = orange,
             anchorcolor = ForestGreen,
]{hyperref} 

\usepackage{amsmath}
\usepackage{amssymb}
\usepackage{amsthm}
\usepackage{bbm}
\usepackage{epigraph}

\usepackage{thmtools} 

\declaretheoremstyle[shaded={rulecolor=gray,rulewidth=1pt, bgcolor={rgb}{1,1,1}}]{boxed}

\declaretheoremstyle[spacebelow=\parsep, spaceabove=\parsep]{unboxed}

\declaretheorem[style=boxed]{lemma}

\declaretheorem[style=boxed,sibling=lemma]{theorem}

\newcommand*{\cM}{\mathcal{M}}

   \newcommand{\con}[1]{\operatorname{Con}(#1)}

	\renewcommand{\vec}[1]{\mathbf{#1}}

	\newcommand{\abs}[1]{\left| #1 \right|}

\begin{document}

\title{Contextuality without access to\newline a tomographically complete set}

\author{Matthew F.\ Pusey}
\affiliation{Department of Computer Science, University of Oxford,  Wolfson Building,  Parks Road, Oxford OX1 3QD, UK}

\author{L\'idia del Rio}
\affiliation{Institute for Theoretical Physics, ETH Zurich, 8093 Z\"{u}rich, Switzerland}

\author{Bettina Meyer}
\affiliation{Institute for Theoretical Physics, ETH Zurich, 8093 Z\"{u}rich, Switzerland}
\affiliation{Niels Bohr Institute, University of Copenhagen, Blegdamsvej 17, 2100 Copenhagen, Denmark}

\date{17th April 2019}

\begin{abstract}
The non-classicality of single quantum systems can be formalised using the notion of contextuality. But can contextuality be convincingly demonstrated in an experiment, without reference to the quantum formalism? The operational approach to contextuality due to Spekkens requires finding operationally equivalent preparation procedures \cite{Spekkens2005}. Previously these have been obtained by demanding indistinguishability under a set of measurements taken to be tomographically complete. In the language of generalised probability theories, this requires the ability to explore all the dimensions of the system's state space. However, if the true tomographically complete set is larger than the set assumed, the extra measurements could break the operational equivalences and hence eliminate the putative contextuality. Such extra dimensions could arise in post-quantum theories, but even if quantum theory is exact there can be unexpected degrees of freedoms due to imperfections in an experiment. Here we design tests of contextuality that are immune to this effect for a given number of extra measurements in the tomographically complete set, even if nothing is known about their statistics. This allows contextuality to be demonstrated with weaker assumptions.
\end{abstract}

\maketitle
\hypersetup{pdftitle={Contextuality without access to a tomographically complete set}}


\setlength{\epigraphwidth}{4in}
\epigraph{\ldots as we know, there are known knowns; there are things we know we know. We also know there are known unknowns; that is to say we know there are some things we do not know. But there are also unknown unknowns: the ones we don't know we don't know.}{Donald Rumsfeld, Department of Defense news briefing, 2002}

\section{Introduction}
\label{sec:introduction}
 
 Many of the ways in which quantum theory departs from classical intuitions can be formalised and unified by the notion of \emph{contextuality} \cite{kochen1968,bell1966}.
 To understand contextuality, it is helpful to take Spekkens' view \cite{Spekkens2005}, which defines a noncontextual model as one where whenever two operational procedures  (like pressing some sequences of buttons in the lab) produce indistinguishable results, then the model should describe the procedures identically at the ontological level. 
 If this is not possible, so that the model is proven to have a fundamental degeneracy at the ontological level, it is said to be contextual. 

However, even the most careful experimental demonstration of such contextuality to date \cite{Mazurek2016} relies on a key assumption: that a tomographically complete set of preparations and measurements were achieved.\footnote{In the language of generalised probabalistic theories, this requires knowing the dimension of the state space.} This allows the determination of which procedures are indistinguishable. If the assumption is false, then operational equivalences may evaporate and the experiment may admit a noncontextual model.

For a simple example with a qubit, imagine an experimenter was only aware of the measurements represented by real-valued POVM elements. In other words, the experimenter believes the Pauli $X$ and $Z$ measurements form a tomographically complete set. Then the experimenter would believe preparations correpsonding to the two eigenstates of the Pauli $Y$ are operationally equivalent. If the two preparations could be shown to differ at the ontological level, the experimenter would claim to have demonstrated contextuality. But this conclusion would be mistaken, because the two preparations are not in fact operationally equivalent, indeed a measurement of Pauli $Y$ can distinguish them perfectly. In more realistic examples, the experimenter will have a firm grasp of quantum theory but may have underestimated the dimension of their system (perhaps due to unexpected non-Markovian interaction with the environment), or quantum theory might not provide an exact description at all.

To ameliorate this difficulty, we will develop tests of contextuality that still work even if there are a certain number of unknown procedures in the tomographically complete set.

In particular, we focus on preparation contextuality and consider a scenario where the true tomographically complete set consists of some of the measurements which can actually be performed (Rumsfeld's "known knowns") but also a known number of measurements whose statistics remain completely unknown ("known unknowns"). This prevents identifying any two given preparations as operationally equivalent, because they may differ on the unknown measurements. Nevertheless, we find that we can sometimes prove that operationally equivalent preparations must exist (which preparations they are depends on the statistics of the unknown measurements), and are able to ground a proof of contextuality on such partially-characterised equivalences.

The remaining loophole is of course that the number of unknown measurements may in fact be larger than assumed (the excess being the dreaded "unknown unknowns"). But even the interpretation of ``loophole free'' Bell experiments depends on discounting certain logical possibilities (such as correlations between measurement settings and hidden variables) on the grounds of physical implausibility \cite{pironio2015}. The potential power of our results is that any analogous plausibility arguments for contextuality experiments need only provide some bound on the total number of measurements in a tomographically complete set, whereas before our results it seemed necessary to rule out the existence of any measurements at all that were not characterised by the measurements that were actually done.

Our main technical results are as follows:
\begin{itemize}
	\item We show in section~\ref{subsec:nofinite} that for any finite number of additional measurements, there is a proof of contextuality (using known measurements as on a qubit) that works regardless of the statistics of the unknown measurements. However, this proof uses a large number of preparations and measurements that must be almost noiseless.
	\item For the simplest case of one unknown measurement, in section~\ref{subsec:simplestcase} we construct a more experimentally-friendly proof that uses the minimal numbers of preparations and measurements.
	\item For any finite number of additional measurements, in section~\ref{subsec:algo} we provide an algorithm able to confirm that contextuality can be shown using a given set of statistics for the known measurements.
\end{itemize}

\section{Setting and framework}
\label{sec:setting}

\subsection{Operational description of experiments}

An  operational theory is defined by indexing all the procedures that an experimenter could (in principle) implement. These can include: a set of
 measurements $\mathcal M$, where each measurement $M \in \mathcal M$ has associated outcomes $\mathcal K_M$;  preparations $\mathcal P$; other transformations.
Note that there is not necessarily a fundamental physical difference between preparations, transformations and measurements; the distinction may just be a practical one.

The underlying causal model behind a prepare-and-measure protocol assumes that the choice of preparation is independent of the choice of measurement, and that the outcome may depend on both,
\begin{align*}
     \left.\begin{aligned}
        \mathcal P  \\
        \mathcal M 
       \end{aligned}
 \right\}
 \xrightarrow{ \  \ }
 \mathcal K_M.
\end{align*}

\subsubsection{Probabilistic descriptions}

In order to make probabilistic statements, agents must assume that they can implement the same procedure independently many times. In this case we  lift the observed measurement frequencies to probability distributions $P^e(k|P,M)$,\footnote{$P^e$ stands for experiment/empirical probabilities.} as $\mathcal K_M$ becomes associated with a random variable $K_M$,
\begin{align*}
     \left.\begin{aligned}
        \mathcal P  \\
        \mathcal M 
       \end{aligned}
 \right\}
 \xrightarrow{ \ P^e(k|P, M) \ }
 \mathcal K_M.
\end{align*}
For the purpose of this work it is also helpful to consider settings where the experimenter may have access to an external source of randomness that enables arbitrary probability distributions $Q(P)$ to be used by the experimenter to choose a preparation.
\begin{align*}
     \left.\begin{aligned}
      \star \xrightarrow{\quad Q(P) \quad } \mathcal P  \\
        \mathcal M 
       \end{aligned}
 \right\}
 \xrightarrow{ \ P^e(k|P, M) \ }
 \mathcal K_M,
\end{align*}
so we have a convex combination of the original $P^e$,
$$P^{e}(k|P_Q,M)=  \sum_P Q(P)\  P^e(k|P,M)  .$$

\subsubsection{Operational equivalence}

We say that two preparations are operationally equivalent if we cannot distinguish them through the measurements specified by the experiment. Here we will be particularly interested in distributions $Q(P)$ and $Q'(P)$ leading to operationally equivalent preparations $P_Q$ and $P_{Q'}$:
\begin{align}
	P^{e}(k|P_Q,M) &=  P^{e}(k|P_{Q'},M) \
    \Leftrightarrow \nonumber \\ 
    \sum_P Q(P)\ P^e(k|P,M) &= \sum_P Q'(P)\ P^e(k|P,M),
    \label{eq:operational_equivalence}
\end{align}
for all $M\in \mathcal M$ and $k \in \mathcal K_M$.

\subsection{Tomographic completeness}

The notion of a \emph{tomographically complete set of measurements} is also an operational one, and relative  to a set of preparations $\mathcal P$. The idea is that there exists a subset of measurements $\mathcal M_C \subset \mathcal M $  which determine the statistics of all the other measurements in $\cM$. That is, for any measurement $M \in \mathcal M$, there exists a deterministic function $f_M$ satisfying
\begin{align*}
P^e(k|P,M)= f_M( k, \{ P^e(k'|P, M')\}_{k' \in \mathcal K_{M'}, M' \in \mathcal M_C}  ),
\end{align*}
for all outcomes $ k \in \mathcal K_M$ and preparations $P \in \mathcal P$. To ensure the right behaviour under convex mixtures, $f_M(k|\cdot)$ must be a linear function \cite{hardy2001,barrett2007}.\footnote{Analogously, given a fixed set of measurements $\mathcal M$,  an operational theory may have a notion of \emph{tomographically complete set of preparations}. Since this is only needed to find equivalent measurements (as for ``measurement noncontextuality''), we will not need to refer to this concept explicitly.}

For finite fixed sets of preparations, claims that a set of measurements is complete can be experimentally falsified by finding a measurement for which no such $f$ exists.

\subsection{Geometric representation of measurement statistics}

\begin{figure}[t]
    \centering
    {\bf a.}
    \includegraphics[width=0.35\textwidth]{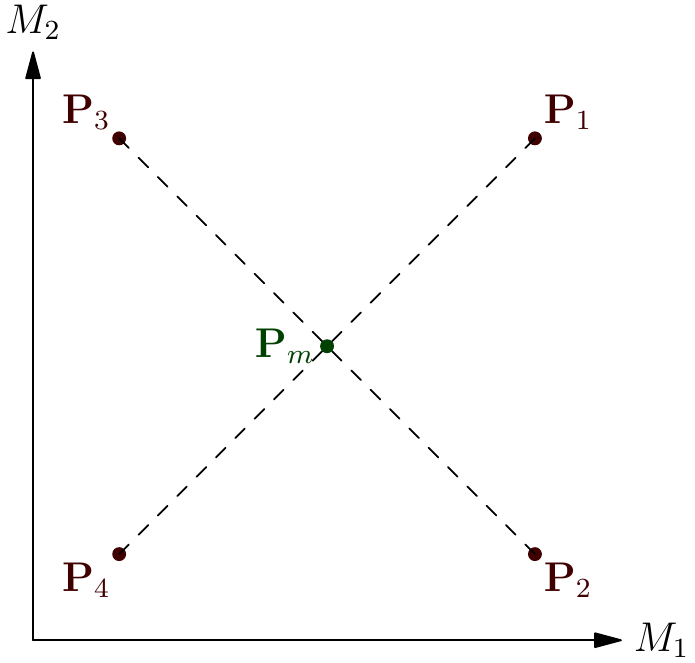}
    \quad
    {\bf b.}
    \includegraphics[width=0.4\textwidth]{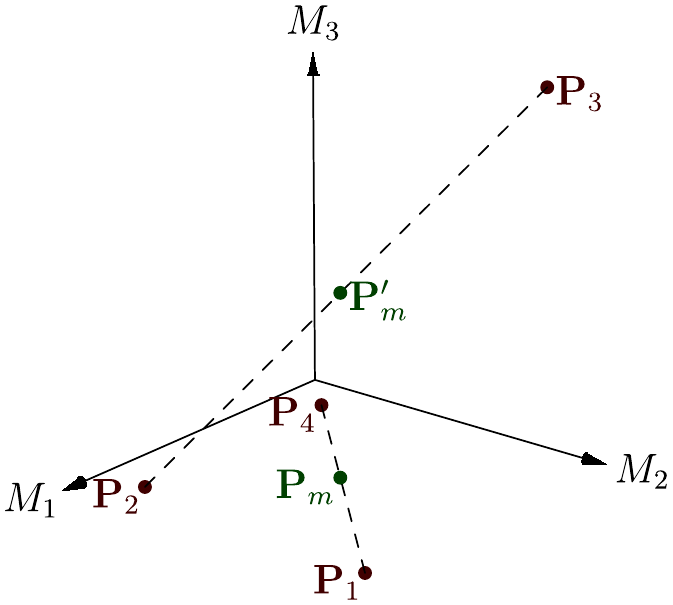}
    \caption{{\bf Operational equivalences.} 
    We represent the measurement statistics of $n$ preparations and $m$ binary measurements  as $n$  points in a $m$-dimensional vector space (Eq.~\ref{eq:statistics_vectors}).   {\bf a.}
    Four preparations and two measurements: the convex hulls of $\{\vec P_1, \vec P_4\}$ and $\{\vec P_2, \vec P_3\}$ (dashed lines) intersect. The point of intersection $\vec P_m$ corresponds to an operational equivalence. {\bf b.} The same preparations, one extra measurement: the operational equivalence was lifted. In other words, measurement $M_3$ allows us to distinguish the two mixtures $\vec P_m$ and $\vec P'_m$.}
    \label{fig:operational_equivalence}
\end{figure}

We consider a setting where $\mathcal P$ consists of $n$ preparations, and $m$ binary measurements form a tomographically complete set.
 The measurement statistics of each preparation $P_i$ may be fully described by the $m$-dimensional vector \cite{hardy2001,barrett2007}
\begin{align}
    \vec P_i = 
\begin{pmatrix}
P^e(0|P_i ,M_1)\\ P^e(0|P_i, M_2) \\ \vdots \\ P^e(0|P_i, M_d)
\end{pmatrix} \in [0,1]^{\otimes m}.
\label{eq:statistics_vectors}
\end{align}
(Since we are considering binary measurements, $P^{e}(1|P, M_j) = 1-P^{e}(0|P, M_j)$, so we only need the probability of outcome 0 to define the space.)
Let us call  the set of all such vectors $P_{\cM}:= \{ \vec P_i\}_i \subset [0,1]^{\otimes m}$. 
Adding a new unknown measurement to $\cM_C$ corresponds to increasing the dimension of the vector space where $P_{\cM}$ lives by one.

The measurement statistics of any convex combination of preparations will lie in the convex hull of $P_{\cM}$, as 
$$  
\begin{pmatrix}
P^{e}(0|P_Q ,M_1)\\ P^{e}(0|P_Q, M_2) \\ \vdots \\ P^{e}(0|P_Q, M_d)
\end{pmatrix} 
= 
\begin{pmatrix}
 \sum_i Q(P_i)\  P^e(k|P_i,M_1) \\ 
 \sum_i Q(P_i)\  P^e(k|P_i,M_2) \\ 
 \vdots \\ 
 \sum_i Q(P_i)\  P^e(k|P_i,M_d)
\end{pmatrix} 
= \sum_i Q(P_i)\ \vec P_i \in \con {P_{\cM}},
$$
where $\con S$ denotes the convex hull of a set $S $ of vectors. 
In particular, tomographic completeness ensures that operational equivalence of two preparation mixtures $Q$ and $Q'$ can be expressed as 
\begin{align}
    \sum_i Q(P_i)\ \vec P_i
    = \sum_i Q'(P_i)\ \vec P_i.
    \label{eq:operational_equivalence_geometric}
\end{align}
This corresponds to finding two subsets of $P_{\cM}$ whose convex hulls intersect, as shown in Fig.~\ref{fig:operational_equivalence}a.

\subsection{Ontological models}

Ontological models form a simple framework for describing physical systems underlying an operational description. In particular, it is assumed that between preparation and measurement the system has a complete description $\lambda$, called the \emph{ontic state}. It is further assumed that the dependence of the measurement outcome on the preparation procedure is mediated by $\lambda$. Hence we consider the causal structure
\begin{align*}
     \left.\begin{aligned}
        \mathcal P  \xrightarrow{ \quad  } \Lambda\\
        \mathcal M 
       \end{aligned}
 \right\}
 \xrightarrow{ \quad }
 \mathcal K_M.
\end{align*}

Ontological models may be non-deterministic: for example, the mapping between descriptions and elements of the underlying theory could be done by probabilistic maps or non-probabilistic embeddings.  
In this work we will consider probabilistic models, of the form
\begin{align*}
     \left.\begin{aligned}
        \mathcal P  \xrightarrow{ \quad \mu(\lambda|P) \quad } \Lambda\\
        \mathcal M 
       \end{aligned}
 \right\}
 \xrightarrow{ \quad P^t(k|\lambda, M) \quad }
 \mathcal K_M,
\end{align*}
so that the observed statistics can be decomposed as
\begin{align*}
    P^e(k|P,M) = \int_\Lambda d\lambda \ \mu(\lambda|P)   \ 
    P^t(k|\lambda, M).
\end{align*}

When the preparation is chosen according to some $Q(P)$ we obtain the model 
\begin{align*}
     \left.\begin{aligned}
        \star \xrightarrow{ \ Q(P) \quad }   \mathcal P  \xrightarrow{ \ \mu(\lambda|P) \quad } \Lambda\\
        \mathcal M 
       \end{aligned}
 \right\}
 \xrightarrow{ \ P^t(k|\lambda, M) \quad }
 \mathcal K_M,
\end{align*}
which results in 
\begin{align*}
    P^{e} (k|P_Q,M)
    &=  \sum_P Q(P)\  P^e(k|P,M)  \\
    &=  \sum_P Q(P)\  
    \int_\Lambda d\lambda \ \mu(\lambda|P)   \ 
    P^t(k|\lambda, M).
\end{align*}

\subsection{Noncontextuality}

The assumption of \emph{noncontextuality} \cite{Spekkens2005} can be summarised as ``operational equivalence implies ontological equivalence.'' Here we focus on \emph{preparation noncontextuality}, which applies this to the case of operationally equivalent preparations.
We start from operational equivalence between two distributions $Q$ and $Q'$, 
\begin{align*}
  P^{e}(k|P_Q,M) &=  P^{e}(k|P_{Q'},M) \
    \Leftrightarrow \\ 
    \Leftrightarrow \
    \sum_P Q(P)\ P^e(k|P,M) &= \sum_P Q'(P)\ P^e(k|P,M) 
    \  \Leftrightarrow \\ 
    \Leftrightarrow \
    \sum_P Q(P)\  
    \int_\Lambda d\lambda \ \mu(\lambda|P)   \ 
    P^t(k|\lambda, M)
    &= \sum_P Q'(P)\  
    \int_\Lambda d\lambda \ \mu(\lambda|P)   \ 
    P^t(k|\lambda, M),
\end{align*}
for all $M\in \mathcal M$ and $k \in \mathcal K_M$. 
Noncontextuality is the assumption that this implies 
\begin{align}
    \sum_P Q(P)\  
     \mu(\lambda|P)   
    &= \sum_P Q'(P)\  
    \mu(\lambda|P) ,
    \label{eq:ontological_equivalence}
\end{align}
for all $\lambda \in \Lambda$. Notice that \eqref{eq:ontological_equivalence} always implies the operational equivalence. The justification for noncontextuality is that ontological equivalence is the best explanation for operational equivalence: the preparations cannot be distinguished because the resulting systems have identical properties.

\section{Results}
\label{sec:result}
\subsection{Robust proofs of contextuality need additional preparations}\label{subsec:needadditionalprep}
Before constructing proofs of contextuality that work in the face of unknown measurements, we first show that a price must be paid in terms of the number of (known) preparations.
\begin{theorem} Consider a setup with $n$ preparations and $m$ known measurements, such that if the known measurements are tomographically complete the scenario does not admit a noncontextual model. Suppose all $n$ preparations are crucial to the contextuality: any $n-1$ of them alone admit a noncontextual model. If the true tomographically complete set in fact includes an additional binary measurement $M_*$, then there exists an assignment of probabilities to $M_*$ allowing a noncontextual model for all $n$ preparations.
\end{theorem}

\begin{proof} The assignment is as follows: $M_*$ always returns the first outcome for $n-1$ of the preparations and always returns the second outcome for the final preparation.
  
We can now construct a noncontextual model as follows. By assumption there exists a noncontextual model with ontic state space $\tilde \Lambda$ for the original measurements on the $n-1$ preparations. We can then specify that $M_*$ gives the first outcome for any ontic state in that model:
\begin{equation}
  P^t(k|\lambda, M_*) = \delta_{k0} \qquad \forall \lambda \in \tilde\Lambda.
\end{equation}
This model therefore reproduces the correct operational predictions for any of the measurements, including $M_*$, on the first $n-1$ preparations. Next, supplement this model with an additional ontic state $\lambda_*$ that is prepared with certainty by the final preparation $P_n$, i.e. $\mu(\lambda_*|P_n) = 1$. Set
\begin{equation}
  P^t(k|\lambda_*, M) = P^e(k|P_n, M),
\end{equation}
so that the operational probabilities for $P_n$ are trivially reproduced.

Now we argue that this extended model, with state space $\Lambda = \tilde \Lambda \cup \{ \lambda_* \}$, is noncontextual. Suppose two distributions $Q$ and $Q'$ give rise to operationally equivalent mixtures: $\sum_i Q(P_i) \vec P_i = \sum_i Q'(P_i)\vec P_i$. Then in particular, they are equivalent on the second outcome of $M_*$:
\begin{equation}
  \sum_i Q(P_i) P^e(1|P_i,M_*) = \sum_i Q'(P_i) P^e(1|P_i,M_*),
\end{equation}
and recalling that by construction $P^e(1|P_i,M_*) = \delta_{in}$ this gives $Q(P_n) = Q'(P_n) := q$. Except in the trivial case $Q = Q'$, we must have $q < 1$. Subtracting $q\vec P_i$ from each side of $\sum_i Q(P_i) \vec P_i = \sum_i Q'(P_i)P_i$ we find
\begin{equation}
  \sum_{i=1}^{n-1} Q(P_i) \vec P_i = \sum_{i=1}^{n-1} Q'(P_i)\vec P_i,
\end{equation}
and dividing through by $1 - q$ gives a normalized operational equivalence. Hence we can apply the noncontextuality of the model for the first $n-1$ preparations to conclude that 
\begin{equation}
  \sum_{i=1}^{n-1} Q(P_i) \mu(\lambda|P_i) = \sum_{i=1}^{n-1} Q'(P_i) \mu(\lambda|P_i) \qquad \forall \lambda \in \tilde\Lambda.
\end{equation}
By construction $\mu(\lambda|P_n) = 0$ for all $\lambda \in \tilde\Lambda$ and so we in fact have 
\begin{equation}
  \sum_{i=1}^{n} Q(P_i) \mu(\lambda|P_i) = \sum_{i=1}^{n} Q'(P_i) \mu(\lambda|P_i) \qquad \forall \lambda \in \tilde\Lambda.
\end{equation}
Finally, $\mu(\lambda_*|P_i) = \delta_{in}$, so
\begin{equation}
  \sum_{i=1}^{n} Q(P_i) \mu(\lambda_*|P_i) = Q(P_n) = Q'(P_n) = \sum_{i=1}^{n} Q'(P_i) \mu(\lambda_*|P_i),
\end{equation}
and we have established ontological equivalence for all $\lambda \in \Lambda$.
\end{proof}

Iterating this result means that every time we add an additional unknown measurement to the tomographically complete set, we must add at least one more preparation to have any hope of proving contextuality.

\subsection{Robust proofs of contextuality may need additional measurements}\label{subsec:needadditionalmeas}
We also find there can be a price in terms of known measurements. Specifically, if for some $m$ we want a proof of contextuality that works for a tomographically complete set of size $2^m - 1$ (or greater), then at least $m+1$ known measurements are required.

\begin{theorem}
Consider a setup with $m$ known binary measurements. If the true tomographically complete set contains at least $2^m - 1$ binary measurements, then there exists an assignment to the unknown measurements allowing a noncontextual model.
\end{theorem}
\begin{proof}
Consider the trivial ontological model with $2^m$ ontic states where each ontic state is a deterministic assignment to measurement outcomes (i.e. a function from measurements to outcomes), and the distribution over the ontic states is just the product of the $P^e$ for each measurement:
\begin{equation}
  \mu(\lambda|P) = \prod_M P^e(\lambda(M)|P,M).
\end{equation}

Consider the $2^m-1$ (possibly unknown) measurements $\{M_\lambda\}$ that ask ``is the ontic state $\lambda$ or not?'' for the first $2^m-1$ ontic states:
\begin{equation}
  P^e(0|P,M_{\lambda}) = \mu(\lambda|P).
\end{equation}
Notice that the $m$ known measurements are linear functions of the statistics of the $\{M_\lambda\}$, because they simply ask which of two subsets the ontic state is in. Hence we can take the $2^m-1$ measurements $\{M_\lambda\}$ to be our tomographically complete set, and then the model is clearly preparation noncontextual because the operational statistics of these measurements (plus normalisation) uniquely determine the distribution over ontic states.
\end{proof}

\subsection{General result: no finite number of additional measurements allows a noncontextual model of a qubit}\label{subsec:nofinite}
We now turn to the main problem: constructing a proof of contextuality that allows for unknown measurements in the tomographically complete set. First note that by the results of Ref.~\cite{Schmid2017}, if there exists a noncontextual model for a finite number of preparations and measurements then there exists a model with a finite number of ontic states. Hence the following reformulation of the problem will be useful.
\begin{lemma}\label{lemma:dimensionformulation}
Suppose $\Lambda$ is finite. Associate each preparation $P_i$ with a vector $\vec \mu_i = (\mu(\lambda_1 | P_i), \mu(\lambda_2 | P_i), \dotsc)$. Let $k$ be the affine dimension of $\{\vec \mu_i\}$. Then there must be at least $k$ measurements in a tomographically complete set in order for the model to be preparation noncontextual.
\end{lemma}
\begin{proof}
By preparing the corresponding convex combinations of the $\{P_i\}$, we can prepare any distribution over ontic states in the convex hull of the $\{\vec \mu_i\}$. By preparation noncontextuality, each member of the convex hull must correspond to operationally inequivalent preparations, i.e. give different predictions for at least one measurement in the tomographically complete set. Since mapping from $\vec \mu_i$ to measurement probabilities is linear, each measurement can only distinguish the vectors in one direction, and so we need at least $k$ measurements. If one of those measurements is a linear combination of the others, it won't provide a new direction, so the $k$ measurements must be linearly independent. Hence the size of a tomographically complete set of measurements, which spans the set of measurements, must be at least $k$.
\end{proof}

We can obtain one bound using the following rather trivial lemma.
\begin{lemma}\label{lemma:binarymatrixrank}
Let $M$ be a $(0,1)$-valued matrix with rank $k$. Then the number of distinct rows in $M$ is at most $2^k$.\end{lemma}
\begin{proof}
Take $k$ columns that span the columns of $M$. In each row there are $2^k$ possible entries in those columns. Take two rows that have the same entries in those columns. Since the other columns are just linear combinations of the spanning columns, they will also have the same entries in those two rows.\end{proof}

\begin{theorem}For any $k \in \mathbb{N}$ there exists $2^k$ preparations and measurements, with statistics compatible with a qubit model, that would require $k$ measurements in a tomographically complete set for a preparation noncontextual model.
\label{theorem:baggy}\end{theorem}
\begin{proof}
We use essentially the same properties of a qubit as \cite{Hardy2004}. Take $n=2^k$ non-orthogonal pure states as the preparations. Let the measurements be the projections onto those states. For each preparation $P_i$ (with corresponding measurement $M_i$), form a vector of length $\abs{\Lambda}$ whose components $v_j$ are $1$ where $P^t(0|\lambda_j, M_i) = 1$ and $0$ otherwise. $\vec \mu_i$ must be zero wherever $\vec v$ is zero, because projection onto the same state is guaranteed to succeed. Meanwhile the $\vec \mu_{i'}$ with $i' \neq i$ must have some non-zero components where $\vec v$ is zero, because projection onto a different state fails some of the time. Hence the $\{\vec \mu_i\}$ vary in the direction of $\vec v$.

Now consider a new value of $i$, and form the corresponding $\vec v$. This must be distinct from the previous one, because it must be $1$ wherever $\vec \mu_i$ is non-zero, whereas, as argued above, the previous one had at least zero component where $\mu_i$ is non-zero.

Taking the $\vec v$ for all values of $i$ as rows of a matrix $M$ gives a $2^k \times \abs{\Lambda}$ $(0,1)$-valued matrix where every row is distinct. Hence by Lemma~\ref{lemma:binarymatrixrank} the rank of $M$ is at least $k$. So the $\{\vec \mu_i\}$ vary in at least $k$ linearly independent directions, i.e. have an affine dimension at least $k$ and so Lemma~\ref{lemma:dimensionformulation} applies.
\end{proof}

Notice that the dependence of known preparations and measurements on unknown measurements is exponential. The argument stated above requires certain probabilities (the projection of a state onto itself) to be exactly $1$, but in Appendix~\ref{appendix:noisybaggy} we give a version that works provided probability of the projection onto a state failing is $\epsilon < \frac{1}{4}\eta^2$, where $\eta$ is the smallest probability of a projection onto a different state failing (this will shrink with $k$ because some pairs of states will be close together, and it may be smaller than the quantum prediction due to noise).

\subsection{Example: a more economical proof for the simplest case}\label{subsec:simplestcase}

As argued in Ref.~\cite{Pusey2015}, the simplest scenario for a standard proof of preparation contextuality is when two binary measurements form a tomographically complete set, and there are four preparations. Hence we now consider adding one unknown measurement to this scenario. Our aim is then to find a proof of contextuality that is still valid regardless of the statistics of the unknown measurement. Since all four preparations are essential, the result of Section~\ref{subsec:needadditionalprep} tells us we will need to consider an additional (known!) preparation. Similarly, the result of Section~\ref{subsec:needadditionalmeas} tells us that two known binary measurements requires at most $2^2 - 1 = 3$ binary measurements in the tomographically complete set to guarantee a noncontextual model, whereas we want to show at at least 4 are required. Hence we will need an additional known measurement. It will turn out we can make do with an additional known measurement whose statistics are a function of the first two known ones, so this additional known measurement is not required for a tomographically complete set.

In short, we consider here a scenario with five preparations with three binary measurements in the tomographically complete set -- two known and one unknown. We also consider a third known measurement, not part of the tomographically complete set.

The preparations will approximately be pure states evenly spaced around the edge of the rebit Bloch circle, thus forming a pentagon that is approximately regular (Fig.~\ref{fig:pentagon}). By applying a simple geometrical argument, Lemma~\ref{lemma:nocturnal} in Appendix~\ref{appendix:geometry}, we see that, even with the unknown additional measurement, the convex hull of some triple of the states must intersect the convex hull of the remaining pair:
\begin{equation}
    p_a P_a + p_b P_b = p_\alpha P_\alpha + p_\beta P_\beta + p_\gamma P_\gamma. \label{eq:linetriangle}
\end{equation}
One of $(P_\alpha,P_\beta)$, $(P_\beta,P_\gamma)$ or $(P_\gamma,P_\alpha)$ must be adjacent vertices of the pentagon, we take the convention that it is $(P_\alpha,P_\beta)$. We can define the probability $p_{\alpha\beta} =p_\alpha + p_\beta$ and a new preparation $P_{\alpha\beta} = \frac{p_\alpha}{p_{\alpha\beta}}P_\alpha + \frac{p_\beta}{p_{\alpha\beta}}P_\beta$ to obtain the operational equivalence
\begin{equation}
    p_a P_a + p_b P_b = p_{\alpha\beta}P_{\alpha\beta}+ p_\gamma P_\gamma.
\end{equation}

We are now back to four preparations and hence we can apply the inequality from Ref.~\cite{Pusey2015}. If that inequality is violated then there is no noncontextual model even with the additional measurement. However, without knowing the statistics of the additional measurement we don't know the equivalence \eqref{eq:linetriangle} and so we need to check the inequality is violated for all of the possible equivalences of this form.

Define $x_i = P^e(0|P_i,M_0)-P^e(1|P_i,M_0)$ and $y_i = P^e(0|P_i,M_1)-P^e(1|P_i,M_1)$. The noncontextuality inequality in Ref.~\cite{Pusey2015}, subject to $(x_a,y_b)$, $(x_\gamma,y_\gamma)$, $(x_b,y_b)$ and $(x_{\alpha\beta},y_{\alpha\beta})$ being the vertices of a convex quadrilateral in clockwise order \cite[Section III]{Pusey2015}, is:
\begin{equation}
    \det{\begin{pmatrix}
    x_a & y_a & x_a + y_a - 1 & 1 \\
    x_\gamma & y_\gamma & -x_\gamma + y_\gamma + 1 & 1 \\
    x_{\alpha\beta} & y_{\alpha\beta} & x_{\alpha\beta} - y_{\alpha\beta} + 1 & 1 \\
    x_b & y_b & -x_b - y_b - 1 & 1
    \end{pmatrix}} \leq 0.
    \label{eq:simpleineq}
\end{equation}

By the definition of $P_{\alpha\beta}$, $x_{\alpha\beta} = \frac{p_\alpha}{p_{\alpha\beta}}x_\alpha + \frac{p_\beta}{p_{\alpha\beta}}x_\beta$ and $y_{\alpha\beta} = \frac{p_\alpha}{p_{\alpha\beta}}y_\alpha + \frac{p_\beta}{p_{\alpha\beta}}y_\beta$. And since determinants are linear in each row, we have
\begin{multline}
    \det{\begin{pmatrix}
    x_a & y_a & x_a + y_a - 1 & 1 \\
    x_\gamma & y_\gamma & -x_\gamma + y_\gamma + 1 & 1 \\
    x_{\alpha\beta} & y_{\alpha\beta} & x_{\alpha\beta} - y_{\alpha\beta} + 1 & 1 \\
    x_b & y_b & -x_b - y_b - 1 & 1
    \end{pmatrix}} 
    =\\
    \frac{p_\alpha}{p_{\alpha\beta}}\det{\begin{pmatrix}
    x_a & y_a & x_a + y_a - 1 & 1 \\
    x_\gamma & y_\gamma & -x_\gamma + y_\gamma + 1 & 1 \\
    x_{\alpha} & y_{\alpha} & x_{\alpha} - y_{\alpha} + 1 & 1 \\
    x_b & y_b & -x_b - y_b - 1 & 1
    \end{pmatrix}} 
     + \frac{p_\beta}{p_{\alpha\beta}}\det{\begin{pmatrix}
    x_a & y_a & x_a + y_a - 1 & 1 \\
    x_\gamma & y_\gamma & -x_\gamma + y_\gamma + 1 & 1 \\
    x_{\beta} & y_{\beta} & x_{\beta} - y_{\beta} + 1 & 1 \\
    x_b & y_b & -x_b - y_b - 1 & 1
    \end{pmatrix}}.
\end{multline}
Hence it suffices to check for a violation of \eqref{eq:simpleineq} in the two extreme cases $P_{\alpha\beta} = P_\alpha$ and $P_{\alpha\beta} = P_\beta$, i.e. we need
\begin{equation} V_i :=
\det{\begin{pmatrix}
    x_a & y_a & x_a + y_a - 1 & 1 \\
    x_\gamma & y_\gamma & -x_\gamma + y_\gamma + 1 & 1 \\
    x_i & y_i & x_i - y_i + 1 & 1 \\
    x_b & y_b & -x_b - y_b - 1 & 1
    \end{pmatrix}} > 0 \label{eq:videf}
    \end{equation}
    for $i = \alpha$ and $i = \beta$.
    
\begin{figure}
    \begin{center}
        \includegraphics{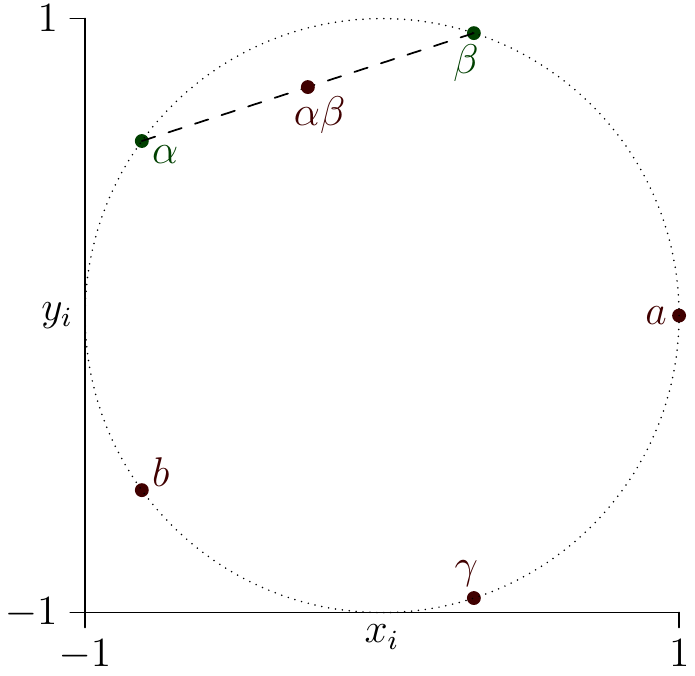}
    \end{center}
    \caption{The ideal statistics as displayed using $(x_i,y_i)$ for the first case considered. The  dotted circle shows all the possibilities for $X$ and $Z$ measurements according to quantum theory (i.e. a slice of the Bloch sphere). The four brown preparations are those to be plugged into the noncontextuality inequality, but it is only known that $P_{\alpha\beta}$ is some convex combination of the known preparations $P_\alpha$ and $P_\beta$ shown in green.}
    \label{fig:pentagon}
\end{figure}    
    
In the ideal case of pure states and projective measurements of the Pauli $X$ and $Z$ matrices, $(x_i,y_i) = (\sin \theta_i, \cos \theta_i)$, with $\theta_i = \left(\frac15 n_i + \frac1{20}\right)2\pi$ and, for example, $(n_a,n_b,n_\alpha,n_\beta,n_\gamma) = (1,3,4,0,2)$, as in Fig.~\ref{fig:pentagon}. Notice that for any values of $p_\alpha$ and $p_\beta$ we have that $(x_a,y_b)$, $(x_\gamma,y_\gamma)$, $(x_b,y_b)$ and $(x_{\alpha\beta},y_{\alpha\beta})$ are the vertices of a convex quadrilateral in clockwise order and so \eqref{eq:simpleineq} applies and we can check its violation using $V_\alpha$ and $V_\beta$ as argued above. For these ideal statistics we obtain
\begin{equation}V_\alpha = \frac{1}{4} \left(5 \sqrt{5}-\sqrt{10 \left(\sqrt{5}+5\right)}+5\right) \approx 1.9 > 0\label{eq:firstValpha}
\end{equation}
and
\begin{equation}V_\beta = \frac{1}{4} \left(5 \sqrt{5}-2 \sqrt{5 \left(2 \sqrt{5}+5\right)}+5\right) \approx 0.6 > 0.
\end{equation} For an actual experiment the two $V_i$ can simply be calculated from \eqref{eq:videf}.

We say ``for example'' above because we have no control over which preparations appear in \eqref{eq:linetriangle}. We can enumerate the possible $a,b,\alpha,\beta,\gamma$ as follows. Firstly $P_a$ can be any of the 5 preparations. $P_b$ must be non-adjacent to $P_a$, which gives a factor of two. However, the difference between $a$ and $b$ is merely conventional, so we can always take the shortest path around the pentagon from $P_a$ to $P_b$ to be clockwise. $P_\gamma$ is then whichever preparation is between $P_a$ and $P_b$ on that path. This leaves two preparations which must be $P_\alpha$ and $P_\beta$, again the assignment is conventional and so we adopt the convention that the shortest path from $P_\alpha$ to $P_\beta$ is clockwise. We can summarise these conventions by saying that $P_a$ can be any of the five preparations, and $(P_a,P_\gamma,P_b,P_\alpha,P_\beta)$ are arranged clockwise. The example of the previous paragraph indeed follows this convention.

\begin{figure}
    \begin{center}
      \includegraphics[width=\textwidth]{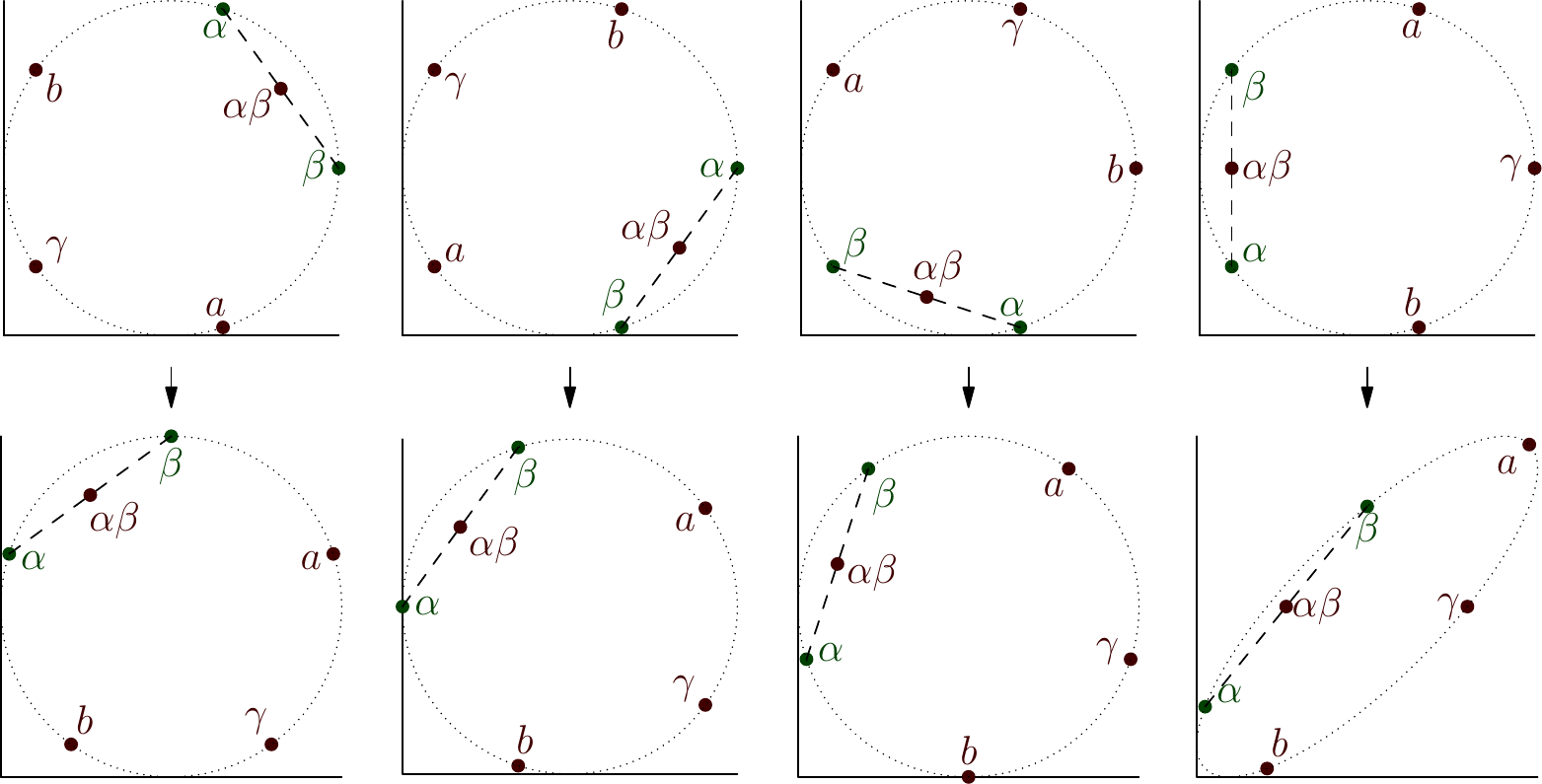}
    \end{center}
    \caption{The other four cases of which equivalence is preserved by the unknown measurement, left to right in the same order as considered in the text. The top row shows the original $(x_i,y_i)$ similarly to Fig.~\ref{fig:pentagon}. The bottom row shows the result of the transformation described for that case in the text. Notice that all the transformation are orientation-preserving and thus the ordering of the relevant quadrilateral's vertices are preserved, whilst the individual vertices are restored to similar positions as in Fig.~\ref{fig:pentagon} to ensure the noncontextuality inequality is violated. Notice also that in the final case the transformation distorts the dotted curve because the two measurements are no longer represented by orthogonal directions in the Bloch sphere.}
    \label{fig:othercases}
\end{figure}    

We have already considered $n_a = 1$ above. Let us now consider the remaining four possibilities, as illustrated in Fig.~\ref{fig:othercases}. If $(n_a,n_b,n_\alpha,n_\beta,n_\gamma) = (2,4,0,1,3)$ we swap the role of the $X$ and $Z$ measurements, and flip the outcome of the $Z$ measurement, so that $(x_i,y_i) = (-\cos\theta_i, \sin\theta_i)$ giving
\begin{equation}
    V_\alpha = 5-\sqrt{\frac{5}{2} \left(5-\sqrt{5}\right)} \approx 2.4 > 0 \label{eq:secondValpha}
\end{equation}
and $V_\beta$ like \eqref{eq:firstValpha}. (This and all the transformations considered below preserve orientation, here because $\det{\begin{pmatrix}0&-1\\1&0\end{pmatrix}}=1 > 0$, and thus the condition on the ordering of the convex quadrilaterals needed for the validity of the argument in the first case is unaffected.)

If $(n_a,n_b,n_\alpha,n_\beta,n_\gamma) = (3,0,1,2,4)$ we can leave $X$ and $Z$ as in the first example, but flip the outcomes of both measurements, so that $(x_i,y_i) = (-\sin\theta_i, -\cos\theta_i)$ giving $V_\alpha$ as in \eqref{eq:firstValpha} and $V_\beta$ like \eqref{eq:secondValpha}.

If $(n_a,n_b,n_\alpha,n_\beta,n_\gamma) = (4,1,2,3,0)$ we can again swap the role of $X$ and $Z$, this time flipping the outcome of the $X$ measurement, so that $(x_i,y_i) = (\cos\theta_i, -\sin\theta_i)$ giving values of $V_\alpha$ and $V_\beta$ the other way round than the first example.

The most difficult case is $(n_a,n_b,n_\alpha,n_\beta,n_\gamma) = (0,2,3,4,1)$. No post-processing of the $X$ and $Z$ measurements gives positive $V_\alpha$ and $V_\beta$. This is where we need the additional known measurement. We now consider a measurement at an angle $\frac3{10}\pi$ along with $Z$ so that $(x_i,y_i) = \left(\sin\left(\theta_i+\frac3{10}\pi\right), \cos\theta_i\right)$ giving 
\begin{equation}
    V_\alpha = \frac{5}{8} \left(3 \sqrt{10-2 \sqrt{5}} - 2 \sqrt{5} -2\right)\approx 0.4 > 0
\end{equation}
and
\begin{equation}
    V_\beta = \frac{5}{4}\left(\sqrt{2 \left(\sqrt{5}+5\right)}-\sqrt{5}-1\right) \approx 0.7 > 0.
\end{equation}

\subsection{An algorithm for checking arbitrary statistics}\label{subsec:algo}

Consider the state space of functions from (known) measurements to outcomes, i.e. deterministic assignments $\lambda(M) = k$.  For a preparation $P$ we define the assignment polytope $\Delta_P$ of distributions consistent with $P$'s statistics, i.e. distributions $\mu(\lambda)$ such that for all measurements $M$:
\begin{equation}
  \sum_{\lambda} \delta_{k \lambda(M)}\mu(\lambda) = P^e(k|P,M).
  \label{eq:deltapconsist}
\end{equation}
Clearly this is nonempty because, we have e.g. $\mu \in \Delta_P$ where $\mu(\lambda) = \prod_M P^e(\lambda(M)|P,M)$. $\Delta_P$ is defined by linear inequalities (positivity) and the linear equality \eqref{eq:deltapconsist}. The first step of the algorithm is to convert this into a list of vertices, using e.g. CDD \cite{Fukuda}.

The second step is, for every pair of disjoint subsets $\{P_{i_1},P_{i_2},\dotsc\}$ and $\{P_{j_1},P_{j_2},\dotsc\}$, to use a simple linear program (based on the lists of vertices) to check if corresponding convex hulls $\con{\{\Delta_{P_{i_1}},\Delta_{P_{i_2}},\dotsc\}}$ and $\con{\{\Delta_{P_{j_1}},\Delta_{P_{j_2}},\dotsc\}}$ intersect.

\begin{theorem}\label{theorem:comput}
If the above algorithm finds there are no intersections, there is no noncontextual model.
\end{theorem}
\begin{proof}
In Ref.~\cite{Schmid2017} it is shown that for a finite number of measurements and outcomes, without loss of generality a noncontextual model can be built from a finite set of ontic states, namely the vertices of the ``noncontextual measurement-assignment polytope''. Since we are not considering any restrictions from measurement noncontextuality, the latter is just the polytope of arbitrary conditional probability distributions, whose vertices are the deterministic assignments.

So suppose there is a noncontextual model for the entire scenario (known and unknown measurements), which, by the above, we can take to use deterministic assignments. We can then coarse-grain this model by identifying ontic states that give the same outcomes for the known measurements (i.e. they differ only in their assignments to unknown measurements), assigning the coarse ontic states probabilities equal to the sum of the probabilities assigned to the corresponding original states. The coarse-grained model will no longer necessarily be able to reproduce the statistics of the unknown measurements, but any preparations that were operationally equivalent for all of the measurements will still have the same distribution in the coarse-grained model. The predictions for the known measurements will still be correct, meaning the distributions over ontic states satisfy \eqref{eq:deltapconsist}.

Suppose we have $n \geq d+2$ preparations, where $d$ is the dimension of the state space including the unknown measurements. Then by Lemma~\ref{lemma:nocturnal} we know that there are disjoint set of preparations whose convex hulls intersect, i.e. some mixture of preparations $p_i P_i$ in the first set is operationally equivalent to some mixture of preparations $q_i P_i$ in the second. Consider the coarse-grained noncontextual  model, with shorthand $\mu_i(\lambda) = \mu(\lambda|P_i)$. Since we know that $\mu_i \in \Delta_{P_i}$, and $\sum_i p_i \mu_i = \sum_i q_i \mu_i$, the convex hulls of the $\Delta_{P_i}$ with $p_i > 0$ must intersect the convex hull of the $\Delta_{P_i}$ with $q_i > 0$. Note that we have to check all the disjoint sets because we don't know which one Lemma~\ref{lemma:nocturnal} will give.
\end{proof}

\section{Conclusions}
\label{sec:conclusions}
\paragraph{Contribution.}
Although determining if two preparations are operationally equivalent requires access to a tomographically complete set, we have shown that the statistics from an incomplete set can be enough to determine that there exist operationally equivalent preparations, and furthermore that we can know enough about that equivalence to be sure that it cannot be represented noncontextually. This shows that the most significant assumption made in previous experimental tests of contextuality can be significantly relaxed. It is not necessary to be able to \emph{perform} a tomographically complete set of measurements, but merely to know how large such a set might be.

We can give a concrete example of a possible application of our results.\footnote{This example was suggested by Rob Spekkens.} Consider a recent paper which applies the concept of tomographic completeness to experimental data \cite{mazurek2017}. An anonymous referee pointed out that there are reasons (e.g. \cite{mueller2013}) to think that for the simplest systems, the number of measurements in a tomographically complete set may be related to the dimension of physical space. This number might not be three. For example, string theory suggests that, despite appearances, the number of spatial dimensions is in fact 9, 10 or 25. The referee argued that the same mechanisms that make those spatial dimensions hard to observe could also make the extra measurements hard to implement (e.g. those measurements would involve rotating systems into the extra dimensions). In short, we could have reasons to suspect that systems which appear to behave like qubits (in particular, 3 binary measurements being tomographically complete) may in fact be described by a theory with up to 25 measurements in the true tomographically complete set. Our results would, in principle, allow a noncontextual model of such a theory to be excluded just using the qubit-like statistics we can already access.

\paragraph{Directions and open questions.}
The number of preparations and known measurements used in the proof of Theorem~\ref{theorem:baggy} is exponential in the number of unknown measurements. The only lower bound we have is from Section~\ref{subsec:needadditionalprep}, namely that at least one extra preparation is required for each unknown measurements, i.e. a linear relationship. It would be good to reduce this gap. In particular, a more efficient proof, perhaps a generalisation of the results in Section~\ref{subsec:simplestcase},  would be useful for experiments and any applications of this form of contextuality.

It would also be useful to clarify the relation between Theorem~\ref{theorem:baggy} and Hardy's proof that there are no ontological models of a qubit using a finite number of ontic states \cite{Hardy2004}. Theorem~\ref{theorem:baggy} is logically a strengthening of Hardy's result because if there was a model with $k$ ontic states then $k-1$ binary measurements in the tomographically complete set would be enough to determine the ontic state and hence ensure a noncontextual model (we used a similar argument in Section~\ref{subsec:needadditionalmeas}). There cannot be a simple argument in the other direction because it is trivial to increase the number of ontic states used in a model without increasing the number of measurements required for a noncontextual model, for example by duplicating every ontic state and splitting the probabilities assigned to each equally. But even though our result is logically stronger, our proof uses exactly the same features of a qubit's statistics as \cite{Hardy2004}. This leads to the question of whether any set of statistics that gives a version of Hardy's result also gives a version of ours. One could ask much the same question for other results that can be viewed as bounds on the number of ontic states required to model some scenario \cite{harrigan2007}, such as classical dimension witnesses \cite{gallego2010,bowles2014} and memory requirements from Kochen-Specker contextuality 
\cite{kleinmann2011,karanjai2018}. If such a translation is possible, it may enable a solution to efficiency problem discussed above.

We only showed in Theorem~\ref{theorem:comput} that our computational algorithm gives a sufficient condition for being able to demonstrate contextuality. It would be useful to understand in what circumstances, if any, the lack of intersections is a necessary condition.

Finally, it is of course crucial for the application of our results to search for more examples of physical reasoning that might bound the number of unknown measurements in an experiment.

\begin{acknowledgements}
We thank Elo\'isa Grifo and Jack Jeffries for the nocturnal lemma and courier services, and Rob Spekkens for discussions. This project began with discussions between LdR and MP at Perimeter Institute for Theoretical Physics, whose support they acknowledge.
LdR also acknowledges support from ERC AdG NLST, from EPSRC grant \emph{DIQIP}, and from the Institute for Quantum Computing at University of Waterloo.
MP also acknowledges support from the
Royal Commission for the Exhibition of 1851.
Research at Perimeter Institute is supported by the government of Canada through Industry Canada and by the Province of Ontario through the Ministry of Economic Development \& Innovation. 

\end{acknowledgements}

\appendix

\section{Geometrical Lemma}
\label{appendix:geometry}
First we recall Carath\'eodory's theorem \cite{Caratheodory1911}. 

\begin{theorem}[Carath\'eodory]
\label{thm:caratheodory}
Let  $\mathcal S$   be a set of points in $\mathbbm R ^d$.  Then every point in the convex hull of $\mathcal S$ is also contained in a simplex whose vertices are at most $d+1$ points of $\mathcal S$, 
$$x \in \con {\mathcal S} \implies \exists\ \mathcal A  \subseteq \mathcal S: \quad x \in \con{\mathcal A} \quad \& \quad  |\mathcal A|\leq d+1.$$
\end{theorem}

This gives the following useful Lemma.
\begin{lemma}
\label{lemma:nocturnal}
Let $\mathcal S = \{ x_i \}_{i=1}^n$  be a set of $n\geq d+2$ points in $\mathbbm R ^d$. 
Then there exist two disjoint subsets of points whose convex hulls intersect, 
$$\exists \  \mathcal A, \mathcal B \subset \mathcal S: \quad
\mathcal A \cap \mathcal B = \emptyset  \quad \& \quad 
\con{\mathcal A} \cap \con{\mathcal B} \neq \emptyset.
$$
Furthermore, $|\mathcal A|, |\mathcal B| \leq d+1$, that is $ \con{\mathcal A} $ and $ \con{\mathcal B}$ are simplices. 
\end{lemma}

\begin{figure}[t]
    \centering
    
    \includegraphics[width=0.4\textwidth]{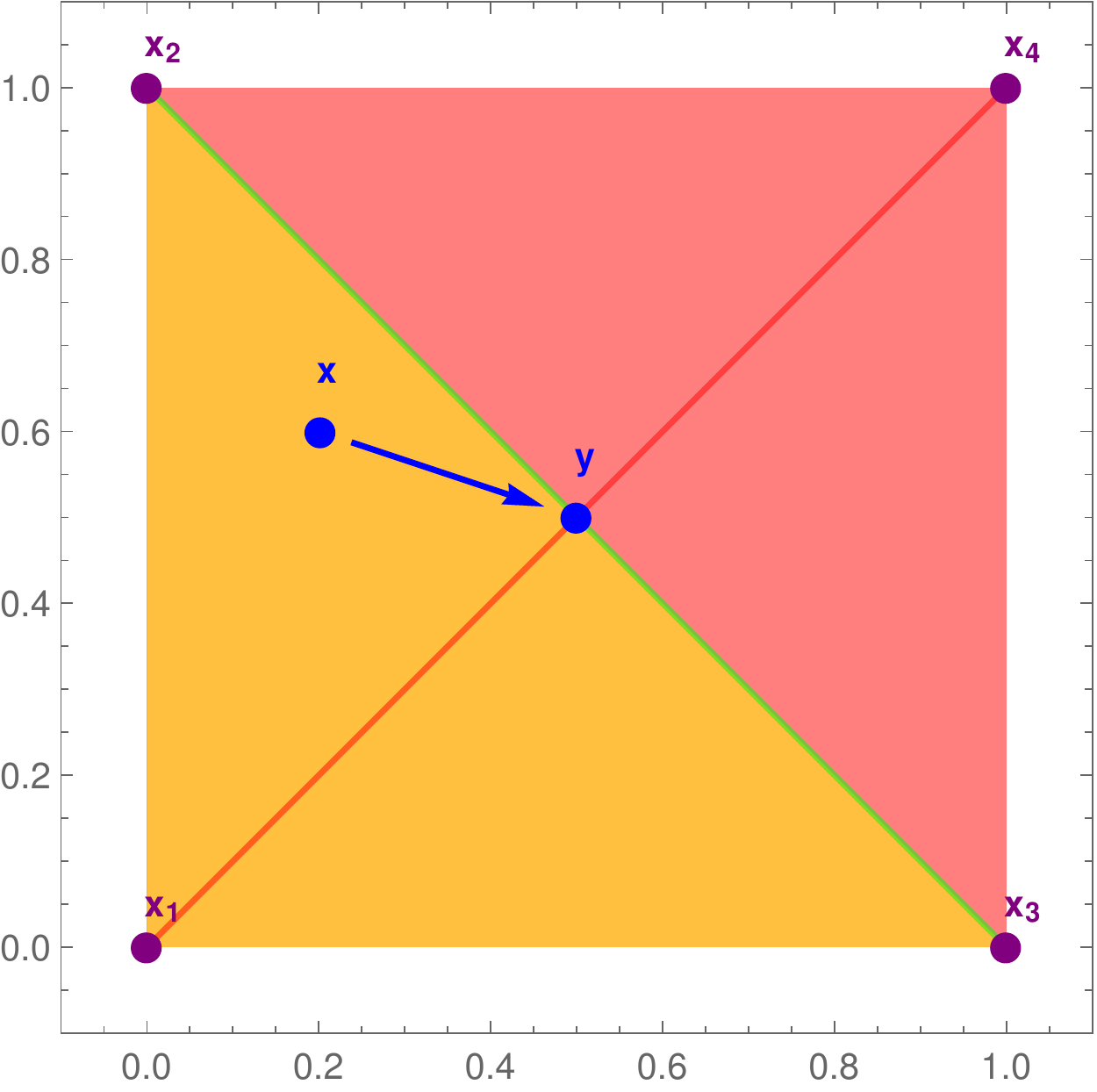}
    \quad
    \caption{{\bf Illustration of Lemma~\ref{lemma:nocturnal}.} 
    We start from a set of four points, $\mathcal S = \{x_1, \dots, x_4\}$ and a point $x= (0.2,\ 0.6)$   in the interior of $\con {\mathcal S}$. This point admits the decomposition $x= 0.3 \ x_1 + 0.5\ x_2 +0.1\  x_3 + 0.1\ x_4$.
    By Carath\'eodory's theorem (Theorem~\ref{thm:caratheodory}), $x$ also admits a decomposition into the yellow simplex, $x=0.2 \ x_1 + 0.6\ x_2 +0.2\  x_3$. The two decompositions are not disjoint, but Lemma~\ref{lemma:nocturnal} allows us to find a new point $y$ with disjoint decompositions. To do so, first we compare the two decompositions and subtract the smallest terms, 
    $(0.3-0.2) \ x_1 + (0.5-0.5)\ x_2 +(0.1-0.1)\  x_3 + (0.1-0.0)\ x_4 = (0.2-0.2) \ x_1 + (0.6- 0.5)\ x_2 +(0.2-0.1)\  x_3
    $, that is $0.1\ x_1 + 0.1\ x_4 =  0.1\ x_2 +0.1\  x_3 $. Next, we renormalize this new point, $y=  (0.1\ x_1 + 0.1\ x_4)/0.2 =(  0.1\ x_2 +0.1\  x_3)/0.2 = (0.5,\ 0.5) $. We found a decomposition into two disjoint simplices, generated by $\{x_1, x_4 \}$ and $\{x_2, x_3 \}$.
    }
    \label{fignocturnal_square}
\end{figure}

\begin{proof}
We take a point in the convex hull of $\mathcal S$,
$x = 1/n \sum_{i=1}^n x_i$. 
By Carath\'eodory's theorem, there exists a subset of at most $d+1$ points whose convex hull contains $x$. We have
$$x = \sum_{i=1}^N \frac{1}n x_i  = \sum_{i=1}^{d+1} b_i x_i  .$$
To find disjoint subsets of extremal points, we subtract $\sum_i \min(1/n, b_i)\ x_i $ from both sides,
$$ \sum_{i: 1/n > b_i} (1/n-b_i) x_i  = \sum_{i: 1/n <b_i} (b_i-1/n) x_i .$$
Now we renormalize to obtain a valid point
\begin{align*}
y &= \frac2{\sum_{i} |1/n-b_i|} 
\sum_{i: 1/n > b_i} (1/n-b_i)\ x_i 
= \frac2{\sum_{i} |1/n-b_i|} \sum_{i: 1/n <b_i} (b_i-1/n)\ x_i .
\end{align*}
Finally we apply Carath\'eodory's theorem to each decomposition, to ensure that we end up with two disjoint simplices. 
\end{proof}

\section[Noise tolerant proof of Theorem 5]{Noise tolerant proof of Theorem \ref{theorem:baggy}}
\label{appendix:noisybaggy}
As before, take $n=2^k$ non-orthogonal pure states as the preparations. Let the measurements be the projections onto those states. For each preparation $P_i$ (with corresponding measurement $M_i$), form a vector $\vec v_i$ of length $\abs{\Lambda}$ whose $\lambda$-th component is $1$ where $P^t(1|\lambda, M_i) < \frac{\eta}2$ and $0$ otherwise.

Now, the sum of the components of $\vec \mu_i$ where $\vec v_i$ is zero (i.e. $P^t(1|\lambda, M_i) \geq \frac\eta2$) must be less than or equal to $\epsilon\frac2\eta$, because projection onto the same state fails with probability at most $\epsilon$. By assumption $\epsilon\frac2\eta < \frac\eta2$. On the other hand for $j \neq i$, the sum of the components of $\vec \mu_{j}$ where $\vec v_i$ is zero must be at least $\frac\eta2$, because projection onto the $i$-th state fails at least $\eta$ of the time, and the terms where $\vec v_i$ is one can contribute at most $\frac\eta2$ to that. Hence the $\{\vec \mu_i\}$ vary in the direction of $\vec v_i$.

By the definition of $\eta$ we have $P^e(1|P_i, M_j) \geq \eta$ for all $j \neq i$. Suppose there was $j \neq i$ with $\vec v_i = \vec v_j$. We have that
\begin{equation}
  \eta \leq P^e(1|P_i, M_j) = \sum_\lambda P^t(1|\lambda,M_j) \mu_i(\lambda).\label{eq:etasum}
\end{equation}
Consider first the terms of the above sum where the corresponding component of $\vec v_i$ is 1:
\begin{equation}
    \sum_{\lambda | v_i(\lambda) = 1} P^t(1|\lambda,M_j) \mu_i(\lambda) \leq \frac{\eta}2  \sum_{\lambda | v_i(\lambda) = 1} \mu_i(\lambda) \leq \frac{\eta}2.
\end{equation}
This leaves the terms whose corresponding component is 0:
\begin{multline}
  \sum_{\lambda | v_i(\lambda) = 0} P^t(1|\lambda,M_j) \mu_i(\lambda) \leq \sum_{\lambda | v_i(\lambda) = 0} \mu_i(\lambda) \leq \sum_{\lambda | v_i(\lambda) = 0} \frac2\eta P^t(1|\lambda,M_i)\mu_i(\lambda) \\\leq \frac2\eta \sum_{\lambda}  P^t(1|\lambda,M_i)\mu_i(\lambda) = \frac2\eta P^e(1|P_i,M_i) \leq \frac{2\epsilon}{\eta},
\end{multline}
where we have used that $P^t(1|\lambda, M_i) \geq \frac\eta2$ for such $\lambda$ and the definition of $\eta$. Putting these back into \eqref{eq:etasum} gives $\eta \leq \frac\eta2 + \frac{2\epsilon}\eta$ and thus $\epsilon \geq \frac14 \eta^2$. Hence if $\epsilon < \frac14 \eta^2$ there cannot be a $j \neq i$ with $\vec v_i = \vec v_j$.

We now conclude as in the noiseless case. Taking the $\{\vec v_i\}$ for all values as rows of a matrix $M$ gives a $2^k \times \abs{\Lambda}$ $(0,1)$-valued matrix where every row is distinct. Hence by Lemma~\ref{lemma:binarymatrixrank} the rank of $M$ is at least $k$. So the $\{\vec \mu_i\}$ vary in at least $k$ linearly independent directions, i.e. have an affine dimension at least $k$ and so Lemma~\ref{lemma:dimensionformulation} applies.

\bibliographystyle{apsrev4-1} 
\bibliography{context}

\end{document}